\documentclass[submission,copyright,creativecommons]{eptcs}
\providecommand{\event}{ACT 2025}
\usepackage{graphicx} 
\usepackage{mathtools}
\usepackage{amsmath,amssymb}
\usepackage{amsthm}
\usepackage{tikz-cd,soul}
\usepackage{enumerate}
\usepackage{enumitem}
\usetikzlibrary{decorations.pathmorphing}

\makeatletter
\let\thetitle\@title
\let\theauthor\@author
\let\thedate\@date
\makeatother

\usepackage{mathtools} 
\usepackage{extarrows}

\newcommand{\mN}{\mathbb N}

\numberwithin{equation}{section}
\numberwithin{figure}{section}
\numberwithin{table}{section}

\theoremstyle{plain} 

\newtheorem{theorem}{Theorem}[section]

\newtheorem{corollary}[theorem]{Corollary}
\newtheorem{lemma}[theorem]{Lemma}

\theoremstyle{definition} 
\newtheorem{definition}[theorem]{Definition}

\newtheorem{example}[theorem]{Example}

\theoremstyle{remark} 
\newtheorem{notation}[theorem]{Notation}
\newtheorem{remark}[theorem]{Remark}

\DeclareMathOperator{\pick}{pick}
\DeclareMathOperator{\Flask}{Flask}

\title{First Steps towards Categorical Algebraic Artificial Chemistry}
\author{Joe Pratt-Johns\thanks{School of Computing, Engineering \& the Built Environment, Edinburgh Napier University, Edinburgh, UK} \thanks{Corresponding author} 
\and Toby {St. Clere Smithe}\thanks{Kodamai Ltd}\and Chris Guiver\footnotemark[1] \and Kevin Hughes\footnotemark[1] \and Peter Andras\footnotemark[1]}
\date{February 2025}

\def\titlerunning{First Steps towards Categorical Algebraic Artificial Chemistry}
\def\authorrunning{J. Pratt-Johns, T. Smithe, C. Guiver, K. Hughes, P. Andras}

\begin{document}

\maketitle
 
\begin{abstract}
We construct a functor that gives a dynamics to an algebraic model of interacting components. The construction generalises a computational model of Fontana and Buss in the field of artificial life known as AlChemy, in which molecules and their chemical interactions are emulated by lambda calculus terms and their application and subsequent reduction. We discuss future directions for the application of category theory to algebraic artificial chemistry as an organisational tool, with a focus on formalising the connection between the algebraic and the dynamical facets of such models.
\end{abstract}



\section{Introduction}
As a branch of artificial life, artificial chemistry seeks to construct models of molecules and their interaction that can be simulated in-silico, with the ultimate goal of observing emergent phenomena that are considered life-like. Since the 1990s, there has been success found in basing the rules of interaction for artificial chemistries on algebraic gadgets~\cite{banzhaf1993self,Fontana1994,Kruszewski2022,Rainford2020}. 

We observe that all such \emph{algebraic artificial chemistries} involve building a dynamical system out of an algebraic structure. We set out to describe a small class of such constructions, under the assumption that category theory is the correct language for doing so. Indeed, we are inspired by constructions of dynamical systems from underlying organisational structures in the applied category theory literature, such as the `grey-boxing' functor of Baez {\em et al.} in the context of reaction networks~\cite{Baez2017}.

The genre of algebraic gadgets that we consider are models (or \emph{algebras}) of single-sorted Lawvere theories~\cite{lawvere1963functorial}. Given such a theory~$\mathcal{T}$ and a morphism~$P \in \mathcal{T}$, we construct a functor \[\Flask^{\mathcal{T}}_P : \mathcal{T}\text{Alg} \to \mathsf{Mark},\]
which assigns a Markov process to a~$\mathcal{T}$-algebra. The construction generalises our main motivating example, the Minimal Chemistry Zero of Fontana and Buss~\cite{Fontana1994, Fontana1996}. 

The~$\Flask$ construction is a proof-of-concept for the idea that category theory may be a useful organisational tool in artificial life research, allowing for formal comparisons between models and providing foundations for generic, compositional code-bases for experimentation. We mention three possible future directions to this end in the final section.

{\bfseries Notation:} We write~$\mathbb{N}$ for the natural numbers, with~$0 \in \mathbb{N}$. A multiset over a set~$S$ is a function~$S \to \mathbb{N}$. The set of finitely supported multisets of~$S$ (i.e. non-zero on finitely many elements) is denoted~$\mathbb{N}^S_{\text{fs}}$.

{
\section{Preliminaries}
In this section, we recall the notions of Lawvere theory and Markov process. In doing so, we define the categories~$\mathcal{T}\mathsf{Alg}$ and~$\mathsf{Mark}$ -- the domain and codomain of the~$\Flask$ functor. 

\subsection{Lawvere theories}
Lawvere theories are certain categories that encode the \emph{syntactic} aspect of an algebraic theory, as well as the axioms of that theory that apply to the syntactic forms. For example, the Lawvere theory for groups encodes all operations that can be performed by some parallel and/or sequential composition of group multiplication, taking inverses, and using the group unit, and also all equations between such operations that hold in every group. In this paper we are concerned only with \emph{single-sorted} Lawvere theories, and we omit this qualifier in what follows. For a detailed treatment, the reader is referred to \cite{borceux1994handbook}.

\begin{definition}
    A \emph{Lawvere theory} is a category with finite products and a certain object~$X$ called the \emph{generic object}, such that every other object is isomorphic to~$X^n$ for some~$n \in \mN$.  
\end{definition}

Each morphism~$X^n \to X^m$ in the Lawvere theory corresponds to some operation that can be performed on~$n$ inputs that produces~$m$ outputs. These include the \emph{structural} morphisms afforded by the cartesian structure that are shared by every theory -- for example, the `swap' morphism~$\langle \pi_2, \pi_1\rangle : X^2 \to X^2$.

Very often, Lawvere theories are defined by giving a finite set, or `basis', of non-structural morphisms, together with a finite list of axioms that they satisfy, in the form of commutative diagrams. The rest of the morphisms of the theory are all those generated by the cartesian structure given such a basis, modulo the axioms. 

\begin{definition}
An \emph{algebra} of a Lawvere theory~$\mathcal{T}$ is a product-preserving functor~$\mathcal{T} \to \mathsf{Set}$. The category~$\mathcal{T}\mathsf{Alg}$ has~$\mathcal{T}$-algebras as objects and natural transformations as morphisms.
\end{definition}

A~$\mathcal{T}$-algebra is specified (up to isomorphism) by a set~$|A| = A(X)$, called the \emph{carrier} of the algebra, as well as an~$n$-ary operation~$|A|^n \to |A|$ for each morphism~$X^n \to X$ in~$\mathcal{T}$. 

Lawvere theories and their algebras respectively provide the notions of syntax and semantics that we use in this paper.

\subsection{Markov processes and their morphisms}
For our purposes, a Markov process is a set~$S$ of \emph{states} and a map~$p$ that assigns to each state~$s$ a finitely supported probability distribution~$p(s)$ on~$S$. We imagine that time passes in discrete steps; if we are in state~$s$ at time~$t$, the distribution~$p(s)$ describes, for each~$s' \in S$, the probability of finding ourselves in state~$s'$ at time~$t+1$.

Another way to say this is that a Markov process is a coalgebra of the (finite) distribution monad on~$\mathsf{Set}$. The category~$\mathsf{Mark}$ will have such coalgebras as objects, and coalgebra homomorphisms as morphisms. We give the explicit definition of~$\mathsf{Mark}$ after recalling the distribution monad.

\begin{definition}
The \emph{distribution monad} on~$\mathsf{Set}$, denoted~$D$, is first, a functor~$\mathsf{Set} \to \mathsf{Set}$, defined as follows.
\begin{itemize}
    \item For a set~$S$, define~$D(S)$ to be the set of functions~$d : S \to [0,1]$ such that 
    \begin{itemize}
        \item~$d(s) > 0$ for only finitely many~$s \in S$ and
        \item~$\sum_{s\in S}d(s) = 1$.
    \end{itemize}
    We call such functions (finite) distributions, and often write them as formal sums~$d = \sum_{s\in S}d(s)s$.
    \item For a map between sets~$f : S \to T$, a distribution~$d \in D(s)$, and an element~$t \in T$, define \[D(f)(d)(t) = \sum_{s \in f^{-1}(t)} d(s).\]
    
\end{itemize}
The monad structure is given by the following natural transformations.
\begin{itemize}
    \item The unit,~$\eta : \operatorname{Id} \to D$, has components~$\eta_S : S \to D(S)$ such that~$\eta_S(s)$ maps~$s$ to~$1$ and maps any~$s' \neq s$ to~$0$. 
    \item The multiplication,~$\mu : D^2 \to D$, has components~$\mu_S : D(D(S)) \to D(S)$ given by \[\mu_S(d)(s) = \sum_{e \in D(S)}e(s)\cdot d(e).\]
    \end{itemize}

\end{definition}
The formal sum notation for distributions is a `good' notation for many reasons. For example, by collecting like terms in the sum, we find that 
\[D(f)(d) = \sum_{t \in T}\left(\sum_{s \in f^{-1}(t)} d(s)\right)t = \sum_{s\in S}d(s)f(s).\]
    %
For the rest of the paper, and in particular in the proof of Theorem~\ref{prop1}, we will use the sum notation for distributions.
\begin{definition}
    The \emph{category of Markov processes},~$\mathsf{Mark}$, is named after its objects. That is, an object of~$\mathsf{Mark}$ is a set~$S$ together with a map~$S \to D(S)$. A morphism from~$p : S \to D(S)$ to~$q : T \to D(T)$ is a map of sets~$f : S \to T$ such that~$q \circ f = D(f) \circ p$.
\end{definition}
}
\section{Algebraic Artificial Chemistry}


{ Dittrich {\em et al.}~\cite{Dittrich2001} identify three constituents into which most artificial chemistries (ACs) can be decomposed: a set of \emph{molecule types}, a collection of \emph{reaction rules}, and a \emph{dynamics}.

Given a set of molecule types~$M$, a reaction rule is a formal graphic~$$a_1 + \cdots + a_n \to b_1 + \cdots + b_k$$where each~$a_i, b_j \in M$.  A pair~$(M,R)$ of a set of molecule types and a set of reaction rules is also known in the literature as a \emph{chemical reaction network}. 

The dynamics of the AC comprise a wide variety of choices that the modeller makes as to how to implement the model, on top of the bare chemical reaction network. This may include in what space the molecules live, how their states are described, and how the system evolves in time. Importantly, the evolution in time is always `based on' the reaction rules, in that the reaction rules can informally be considered a parameter of the final implementation.

There are two ways that reaction rules can be defined: explicitly or implicitly. An explicit definition is simply a finite list of rules: all interactions are known by the modeller \emph{a priori}. In this case, the set of molecule types is often simply a set of symbols with no extra structure. Explicit ACs have seen some interface with category theory already, for example in the form of open Petri nets that interpret chemical reaction networks~\cite{Baez2017}.

An implicit definition is possible when there is extra structure on the molecule types, so that a collection of reaction rules may be given by a formula. In this case, not all possible interactions are known \emph{a priori}. Such ACs are of particular relevance to artificial life because they exhibit \emph{open-endedness}; it is possible to define infinitely many reaction rules with one formula. Consider, for example, the coupling of certain long organic polymers; we cannot write down a finite list of all possible reactions of this kind, yet the outcome of any given coupling may be deduced from the reagents. }

An \emph{algebraic} artificial chemistry (AAC) is an implicit artificial chemistry whose underlying class of molecules and reaction rules come from an algebraic structure. For example, a simple AAC has finite strings as its class of molecules, and a non-finite collection of reaction rules given by the formula \[l_1 + l_2 \to \mathsf{concat}(l_1,l_2),\] for each pair of reagent molecules~$(l_1, l_2)$. This is an AAC, because the set of finite strings under concatenation is a monoid, and the reaction rule is constructed from the monoid operation. As yet, this is not a precise definition; indeed, part of the purpose of this paper is to define algebraic artificial chemistries, or at least a small class of them, in concrete terms.

{
We observe that in this case the set of molecules is given by an algebra of the Lawvere theory for monoids, and that the formula for the reaction rules is encoded by the image of the multiplication morphism under this algebra. All that remains is to give the chemistry its dynamics. Our construction,~$\Flask$, is a certain choice for giving an AAC its dynamics, in which the reaction rules are now \emph{formally} a parameter. In the next section, we repeat this observation with a more interesting AAC.}

\section{AlChemy}

Walter Fontana and Leo Buss~\cite{Fontana1996} were early proponents of using algebraic structures to model molecular interaction and, more generally, the interaction of components of complex emergent systems. Their research program, which involves the study of artificial chemistries arising from various models of computation and even proof-theory, has become known as `AlChemy', and has yielded reproducible results in which self-maintaining organisations of artificial molecules emerge due to selection pressure~\cite{Fontana1994, Mathis2024}. {Recent work from the artificial life community demonstrates a renewed interest in AlChemy. Mathis \emph{et al.}~\cite{Mathis2024} have written a Python wrapper for the original code that has allowed them to reproduce the original results, as well as perform new statistical analyses, with increased rigour and clarity. Kruszewski \emph{et al.}~\cite{Kruszewski2022} study emergent self-reproducing metabolisms in an AlChemy-inspired chemistry based on combinatory logic. Folena \emph{et al.}~\cite{folena2024emergence} study the emergence of enzymes in a similar chemistry. We note that ACs in the mould of AlChemy do not attempt to accurately model or make predictions about real chemistry, rather they are an extreme distillation of the ideas of open-endedness and constructive interaction between components, inspired by the chemistry of living systems.} Our motivating example is the simplest model from the original work on AlChemy: Minimal Chemistry Zero (MC0).

\subsection{Definition of MC0}
{The setting for MC0 is called the `reactor'. A state of the reactor is a multiset of untyped lambda calculus terms. At each discrete time step, the state is updated by colliding a random pair of lambda terms and adding the reduction of their application to the reactor. We give the full description in terms of Dittrich {\em et al.}'s~\cite{Dittrich2001} three constituents.
}

\begin{itemize}
    \item \textbf{Molecules}. The set of molecules,~$M$, is the set of terms in the untyped lambda calculus. That is, choose a countable set of variables,~$V$, and define~$M$ to be the smallest set such that 
    \begin{itemize}
        \item~$V \subset M$ (variables).
        \item~$t \in M$ and~$x \in V \implies \lambda x.t \in M$ (abstraction).
        \item~$t_1, t_2 \in M \implies t_1(t_2) \in M$ (application).
    \end{itemize}
    
    Choose also a deterministic operational semantics of the lambda calculus. This is a partial function~$R \subset M \times M$ called \emph{reduction} that will tell us how to evaluate a term to (hopefully) reach a normal form. By setting a limit for the number of reductions to apply before halting, we obtain a function~$E : M \to M$ that repeatedly applies reductions until there is no reduction to apply, or the limit is reached. 
    \item \textbf{Reaction rules}.  The reaction rules for MC0 are described by the formula \[t_1 + t_2 \to t_1 + t_2 + E(t_1(t_2)).\]
    That is, there is a reaction rule of the above form for each~$(t_1,t_2) \in M^2$. {Note that ``$+$" here is still a formal symbol and not part of the lambda calculus syntax.}
    \item \textbf{Dynamics}. To give MC0 a dynamics based on the reaction rules given by the above formula, first initialise a finite collection of molecules, which we will think of as molecules of a fluid in some `reactor', each with some non-zero kinetic energy that allows them to collide at a constant rate. Time progresses in discrete steps; at each step, do the following.
    \begin{itemize}
        \item Choose two molecules,~$t_1$ and~$t_2$, from the reactor uniformly at random, without replacement. These are the molecules which will `collide'.
        \item Apply the above reaction rule. That is, add~$E(t_1(t_2))$ to the reactor.
        \item Choose a third molecule uniformly at random (this may be~$t_1$,~$t_2$, or~$E(t_1(t_2))$) and remove it from the reactor, thus keeping the total number of molecules constant. 
    \end{itemize}
\end{itemize}

\subsection{Generalising the model}
In order to generalise MC0, we must decide what the `story' of the construction is, which characters in the story should be considered parameters, and in what space each of these parameters should reside. One approach to describe the above construction from a general perspective is explained below. From now on, we will adjust our terminology to incorporate a wider range of modelling contexts: `molecules' will become `components', and a formula for reaction rules will become a `protocol'.
\begin{itemize}
    \item \textbf{Components}. Components are given by a syntax and a semantics. The syntax describes what `kind of thing' the components are, and therefore in what ways they may interact. {The MC0 construction may be repeated with the same reaction rules and dynamics for any sort of components that have an interpretation of `interaction' between a pair of components, so this is what we identify as the syntax in Example~\ref{exmpl_MC0}.} The choice of semantics defines the `mechanics of interaction'; it tells us what the components actually are, and what interaction means. In the case of MC0, interaction `means' application and subsequent reduction. In our telling of this story, we  will take the syntax to be a (single sorted) Lawvere theory,~$\mathcal{T}$, and the semantics to be an algebra of~$\mathcal{T}$. 
    \item \textbf{Protocol}. Given a syntax for components, a protocol is a formula for reaction rules that takes the form \begin{align}\label{formula}t_1 + \cdots + t_k \to s_1 + \cdots + s_l,\end{align}where each~$t_i$ is a meta-variable representing a syntactic term and each~$s_j$ is a syntactic term built from the~$t_i$'s. {An equivalent way to write this formula is as an~$l$-tuple of syntactic formulas in the variables~$t_1,..., t_k$}. Hence, when the syntax is given by a Lawvere theory, this information is encoded by a morphism \[P : X^k \to X^l,\]where~$X$ is the generic object of the theory. The protocol describes what happens when components meet -- that is, it describes in what order and combination they interact.
    \item \textbf{Dynamics}. We consider the reactor of MC0 to be a discrete Markov process. The state of the process at each time step is simply the collection of molecules present in the reactor at that time.
\end{itemize}

{\remark{There is a notational wrinkle that needs 
attention here involving the formal symbol ``$+$". First, we consider the rules~$A + B \to C$ and~$B + A \to C$ to be distinct. This is to allow the order of reagents to have an effect on the dynamics -- for example, when concatenating strings, we need to choose which reagent string goes first. However, note that formula~\ref{formula} defines an identical set of rules regardless of the order of the left-hand-side formal sum, because the variables are quantified over every component type. Conversely, the order on the right-hand-side of reaction rules is often irrelevant when it comes to the dynamics, so that~$A \to B + C$ and~$A \to C + B$ may be considered to be the `same' reaction rule. From this point of view, there is a many-to-one correspondence between protocols and sets of reaction rules. This is borne out in the~$\Flask$ construction below by the fact that~$\Flask^{\mathcal{T}}_{\langle P_1,P_2\rangle} = \Flask^{\mathcal{T}}_{\langle P_2,P_1\rangle}$.
}}

\section{The Flask functor}

In this section, we organise the above story of the construction of MC0 and similar models of interacting components into a functor
\[\Flask^{\mathcal{T}}_P : \mathcal{T}\mathsf{Alg} \to \mathsf{Mark}\] 
that defines a Markov process from an algebraic structure. The functor is parametric in both:
\begin{itemize}
    \item the flavour of algebraic structure to which the components belong (a Lawvere theory~$\mathcal{T}$), which we call the \emph{domain};
    \item the \emph{protocol} (previously `formula for reaction rules'), which is a morphism~$P$ from~$\mathcal{T}$.
\end{itemize}

Continuing the analogy of Fontana and Buss, our construction~$\Flask$ puts terms from an arbitrary algebra of an arbitrary algebraic theory into a simple piece of `glassware', mixes them well, and asks what happens as the terms collide at random according to the protocol. One may read `reactor' instead of `flask' for consistency with the Fontana and Buss vocabulary. 

\subsection{Intuition}
Before giving the definition in full generality, we present a natural-language description of the~$\Flask$ construction in the case that the domain,~$\mathcal{T}$, is the Lawvere theory for groups. In this case, a~$\mathcal{T}$-algebra is simply a group. We should also choose a `protocol' -- a morphism in the theory of groups -- which will dictate what happens when a `collision' occurs in the flask. 

A natural choice for~$P$ is~$\star: X \times X \to X$, the morphism that induces the group multiplication. Taking~$P = \star$ should be understood as taking the following formula for reaction rules as the protocol that gives rise to the dynamics in the flask:~$$g_1 + g_2 \to g_1g_2.$$
In other words, when two group elements collide, they are removed from the system, and are replaced by the single term corresponding to their group product.

We may now describe the action of~$\Flask^\mathcal{T}_P$. For a group~$G$,~$\Flask_P^\mathcal{T}(G)$ is a Markov process whose state space is the set of finitely supported multisets of over~$G$. In the chemistry analogy, a state is a collection of `molecules' labelled by elements of~$G$. The transition probabilities can be described by the following probabilistic procedure for generating a new state~$\tau$ from an old state~$\sigma$:
\begin{itemize}
    \item Choose two members of the multiset~$\sigma$, uniformly at random, without replacement. They will be labelled by group elements, say~$g_1$ and~$g_2$.
    \item Define~$\tau$ to be the multiset that results from removing the chosen members from~$\sigma$ and adding a member labelled by the group element~$g_1g_2$. 
\end{itemize}

For a group homomorphism~$\psi : G \to H$, we define~$\Flask^\mathcal{T}_P(\psi)$ to be the pushforward 
\[\Flask^\mathcal{T}_P(\psi) = \psi_* : \mathbb{N}^{G}_{\text{fs}} \to \mathbb{N}^{H}_{\text{fs}},\]
given by, for each~$\tau \in \mathbb{N}^{G}_{\text{fs}}$ and~$h \in H$, 
\[ \psi_*(\tau)(h) := \sum_{g \in  \psi^{-1}(h)}\tau(g).\]
In other words,~$\Flask^\mathcal{T}_P(\psi)$ relabels the members of each multiset by applying~$\psi$. That this map is a morphism of Markov processes~$\Flask^\mathcal{T}_P(G) \to \Flask^\mathcal{T}_P(H)$ is not immediate, and the proof boils down to the fact that~$\psi$ is a homomorphism of groups. 
 
\subsection{The Flask functor in full generality}

We now give a formal definition of~$\Flask^\mathcal{T}_P$, extending to an arbitrary Lawvere theory~$\mathcal{T}$ and morphism~$P \in \mathcal{T}$. First, notations are introduced for dealing with multisets.

\begin{notation}Let~$S$ be a set, and let~$\sigma : S \to \mathbb{N}$ be a multiset of members labelled by~$S$. For an element~$s \in S$, we write~$\sigma_s$ for the multiset that results from removing an element labelled by~$s$, if it is supported. So, 
\[
\sigma_s(t) := \begin{cases}\max\{0,\sigma(t) - 1\} & \text{ when } t = s \\ \sigma(t) & \text{ otherwise.}\end{cases}
\]
Similarly, we write~$\sigma^s$ for the multiset that results from introducing a member labelled by~$s$:~
\[
\sigma^s(t) := \begin{cases}\sigma(t) + 1 & \text{ when } t = s \\ \sigma(t) & \text{ otherwise.}\end{cases}
\]
For an arbitrary number of removals followed by introductions, we write~
\[\sigma_{s_1,\dots,s_n}^{t_1, \dots, t_m} := (\cdots((\cdots(\sigma_{s_1})\cdots)_{s_n})^{t_1}\cdots)^{t_m} \quad \forall \: t_i, s_j \in S.\]
Note that the order of removals and introductions in the above definition matters, because we do not allow negative multiplicities in the definition of~$\sigma_s$.
\end{notation}

\begin{notation}
Let~$\sigma : S \to \mathbb{N}$ be a finitely supported multiset and let~${s}\in S^k$ for~$k \in \mN$. We write~$\pick_{{s}}(\sigma) \in \mathbb{R}$ for the probability of picking~${s}$ by choosing~$k$ members from~$\sigma$ uniformly at random, without replacement. So,
\[ \pick_{(s_1, ..., s_k)}(\sigma) = \frac{\sigma(s_1)\sigma_{s_1}(s_2) \cdots \sigma_{s_1, \dots, s_{k-1}}(s_k)}{N(N-1) \cdots (N - k+1)},\]
where~$N = \sum_{s\in S}\sigma(s)$. When $N = 0$, we set $\pick_{s_1, \dots,s_k}(\sigma) = 0$.
\end{notation}

\begin{definition}\label{def}
Let~$\mathcal{T}$ be a Lawvere theory with generic object~$X$. Let~$P : X^k \to X^l$ be a morphism in~$\mathcal{T}$. The \emph{Flask functor}~$\Flask^\mathcal{T}_P : \mathcal{T}\mathsf{Alg} \to \mathsf{Mark}$ with \emph{domain}~$\mathcal{T}$ and \emph{protocol}~$P$ is defined as follows:
\begin{itemize}
    \item For a~$\mathcal{T}$-algebra~$A$ we define the Markov process \begin{align*} \Flask_P^\mathcal{T}(A) : \mathbb{N}^{|A|}_{\text{fs}} &\to  D(\mathbb{N}^{|A|}_{\text{fs}}) \\ \sigma &\mapsto \begin{cases}\sum_{{a} \in |A|^k} \pick_{{a}}(\sigma)\left[\sigma_{{a}}^{A(P)({a})}\right] & \text{when } l > 0 \text{ and } k \leq \sum_a \sigma(a)\\
    \sum_{{a} \in |A|^k} \pick_{{a}}(\sigma)\left[\sigma_{{a}}\right] & \text{when } l = 0 \text{ and } k \leq \sum_a \sigma(a)\\
    \eta(\sigma) & \text{when } k > \sum_a \sigma(a),
    \end{cases}
    \end{align*}
    whose state space is the set of finitely supported maps~$|A| \to \mathbb{N}$. Recall that since~$A$ preserves products,~$A(P)$ may be considered as a map~$|A|^k \to |A|^l$. 
    \item For a morphism of algebras~$f : A \to B$, we define~$\Flask_P^\mathcal{T}(f)$ to be the pushforward of~$f_X$. That is,\begin{align*}\Flask_P^{\mathcal{T}}(f) := (f_X)_* : \mathbb{N}^{|A|}_{\text{fs}} &\to \mathbb{N}^{|B|}_{\text{fs}} \\ \sigma &\mapsto \left[ b \mapsto \sum_{a \in f_X^{-1}(b)} {\sigma}(a)\right]. \end{align*}
\end{itemize}
\end{definition}
\begin{theorem}\label{prop1}
   ~$\Flask_P^\mathcal{T} : \mathcal{T}\mathsf{Alg} \to \mathsf{Mark}$ is a well-defined functor.
\end{theorem}

{The following lemma collects relationships between multisets and pushforwards used in the proof of Theorem~\ref{prop1}.}
{ 
\begin{lemma}\label{lem}
   Let~$f : A \to B$ be a map between sets and let~$\sigma$ be a finitely supported multiset over~$A$. Let~$a \in A$,~$(a_1, \dots a_k) \in A^k$, and~$(b_1, \dots, b_k) \in B^k$. Then
    \begin{enumerate}[label=(\roman*)]
        \item~$f_*(\sigma^a) = (f_*\sigma)^{f(a)}$;
        \item~$f_*(\sigma_a) = (f_*\sigma)_{f(a)}$ whenever~$\sigma(a) \geq 1$;
        \item~$f_*(\sigma_{a_1, ..., a_k}) = (f_*)_{f(a_1), ..., f(a_k)}$ whenever~$\pick_{a_1, ..., a_k}\sigma > 0$;
        \item \[\pick_{b_1, ..., b_k}(f_*\sigma) = \sum_{\substack{a_i \in f^{-1}(b_k)\\i = 1, ..., k}}\pick_{a_1, ..., a_k}(\sigma).\]
    \end{enumerate}
\end{lemma}}
\begin{proof} 
    \begin{enumerate}[label=(\roman*).]
    \item Let~$b \in B$. We split into two cases.

    \textbf{Case 1:}~$f(a) = b$. Then \[f_*(\sigma^a)(b) = \sum_{x \in f^{-1}(b)}\sigma^a(x) = \left(\sum_{x\in f^{-1}(b)}\sigma(x)\right) + 1 = (f_*\sigma)_{f(a)}(b).\]

    \textbf{Case 2:}~$f(a) \neq b$. Then \[f_*(\sigma^a)(b) = \sum_{x \in f^{-1}(b)}\sigma^a(x) = \left(\sum_{x\in f^{-1}(b)}\sigma(x)\right) = (f_*\sigma)_{f(a)}(b) \]
    
    \item Suppose~$\sigma(a) \geq 1$. Then~$\sigma_{a}(a) = \sigma(a) - 1$ and~$\sigma_{a}(x \neq a) = \sigma(x)$. Let~$b \in B$. We split into two cases.
    
    \textbf{Case 1:}~$f(a) = b$. Then, using the above characterisation of~$\sigma_a$, we have \[f_*(\sigma_a)(b) = \sum_{x \in f^{-1}(b)}\sigma_a(x) = \left(\sum_{x \in f^{-1}(b)}\sigma(x)\right) - 1 = f_*(\sigma)(b) - 1.\]
    On the other hand, \[f_*(\sigma)_{f(a)}(b) = \max\{0, f_*(\sigma)(b) - 1\} = f_*(\sigma)(b) - 1,\]because~$\sigma(x) \geq 1$ for at least one~$x$ in the preimage of~$b$. 

    \textbf{Case 2:}~$f(a) \neq b$. Then \[f_*(\sigma_a)(b) = \sum_{x : f(x) = b} \sigma_a(x)= \sum_{x : f(x) = b}\sigma(x) = f_*\sigma(b) = (f_*\sigma)_{f(a)}(b).\]
        
    \item We proceed by induction on~$k$. The base case follows from part (ii), since~$\pick_a\sigma > 1 \implies \sigma(a) \geq 1$. In fact, since~$\operatorname{pick}_{a_1, ..., a_k}\sigma > 0 \implies \sigma(a_k) \geq 1$, the induction step also follows from part (ii): 
    \begin{align*} 
    	f_*(\sigma_{a_1, ... , a_k}) & = f_*((\sigma_{a_1, ..., a_{k-1}})_{a_k}) = (f_*\sigma_{a_1, ..., a_{k-1}})_{f(a_k)} = \\
    				     & = (f_*(\sigma)_{f(a_1), ..., f(a_{k-1})})_{f(a_k)} = f_*(\sigma)_{f(a_1), ..., f(a_k)}.
    \end{align*}				     
    We use the induction assumption in the third equality, together with the fact that~$\operatorname{pick}_{a_1, ...., a_{k}}\sigma > 0 \implies\operatorname{pick}_{a_1, ...., a_{k-1}}\sigma > 0$.
    \item It will be enough to show that~$$f_*\sigma(b_1) \times (f_*\sigma)_{b_1}(b_2) \times \cdots \times (f_*\sigma)_{b_1,...,b_{k-1}}(b_k) = \sum_{\substack{a_i : f(a_i) = b_i\\i = 1, ..., k}}\sigma(a_1)\times\sigma_{a_1}(a_2)\times\cdots\times\sigma_{a_1, ..., a_{k-1}}(a_k).$$
    The quickest way to see that this holds is to consider that~$\sigma$ and~$f_*\sigma$ give two different labels to each token in the same bag -- the ``$A$-labels" and the ``$B$-labels". For each token with~$A$-label~$a$, its~$B$-label is~$f(a)$. Both sides of the equation count the number of ways of collecting labels~$(b_1, ..., b_k)$ from the bag by taking tokens one at a time; the left hand side counts directly, and the right hand side counts by summing the number of ways of collecting labels~$(a_1, ..., a_k)$ for each vector of~$A$-labels whose collection also means that we have collected~$B$-labels~$(b_1, ..., b_k)$.
    \end{enumerate}
\end{proof}

\begin{proof}[Proof of Theorem~\ref{prop1}]
    Let ~$f : A \to B$ be a morphism of~$\mathcal{T}$-algebras. We need to show that~$\Flask^{\mathcal{T}}_P(f) = (f_X)_*$ is indeed a morphism of the relevant Markov processes. That is, we need to show that the square below commutes.
\[\begin{tikzcd} A && {} && {\mathbb{N}^{|A|}_{\text{fs}}} && {D(\mathbb{N}^{|A|}_{\text{fs}})} \\ & {} && {} & {} \\ B &&&& {\mathbb{N}^{|B|}} && {D(\mathbb{N}^{|B|})} \arrow["f"', from=1-1, to=3-1] \arrow["{\Flask_P^\mathcal{T}(A)}", from=1-5, to=1-7] \arrow["{(f_X)_*}"', from=1-5, to=3-5] \arrow["{D((f_X)_*)}", from=1-7, to=3-7] \arrow[squiggly, maps to, from=2-2, to=2-4] \arrow["{\Flask_P^\mathcal{T}(B)}"', from=3-5, to=3-7] \end{tikzcd}\]

Let $\sigma \in \mathbb{N}^{|A|}_{\text{fs}}$. The non-trivial case to check is when $k \leq \sum_a \sigma(a)$. We write~$f := f_X$ and~$f^k := f_{X^k}$. By definition, we have
\begin{align*}\left[D(f_*) \circ \Flask_P^{\mathcal{T}}(A) \right] (\sigma) &= \sum_{{a} \in |A|^k}\pick_{{a}}(\sigma) \left[f_*\left(\sigma_{{a}}^{A(P)({a})}\right)\right]; \\
\left[\Flask_P^{\mathcal{T}}(B) \circ f_* \right](\sigma) &= \sum_{{b} \in |B|^k}\pick_{{b}}(f_*\sigma)\left[\left(f_*\sigma\right)_{{b}}^{B(P)({b})}\right].\end{align*}
For brevity, define \begin{align*}\alpha_{{a}} &:= f_*\left(\sigma_{{a}}^{A(P)({a})}\right);\\\beta_{{b}} &:= \left(f_*\sigma\right)_{{b}}^{B(P)({b})}.\end{align*}
{By Lemma~\ref{lem}, we have
\[\alpha_{{a}} = (f_*\sigma)_{f^k(a)}^{f^l \left[ A(P)({a})\right]} = (f_*\sigma)_{f^k({a})}^{B(P)(f^k({a}))}\]
whenever~$\pick_a\sigma > 0$.} The second equality is due to naturality of~$f$. It follows that {when~$\pick_a \sigma > 0$,} \begin{align}\label{implication}f^k({a}) = {b} \implies \alpha_{{a}} = \beta_{{b}}.\end{align}
{Combining the above with Lemma~\ref{lem} part (iv), we have}
\begin{align*}\sum_{{a}\in|A|^k}\pick_{{a}}(\sigma)\alpha_{{a}} 
&= \sum_{{b} \in |B|^k}\left[\sum_{{a} \in (f^k)^{-1}({b})}\pick_{{a}}(\sigma)\alpha_{{a}} \right]\\ 
&\overset{\eqref{implication}}{=} \sum_{{b} \in |B|^k}\left[\sum_{{a} \in (f^k)^{-1}({b})}\pick_{{a}}(\sigma)\right]\beta_{{b}} \\ 
&=\sum_{{b} \in |B|^k}\pick_{{b}}(f_*\sigma)\beta_{{b}}.\end{align*}
Finally, since composition and identities for~$\mathsf{Mark}$ are the same as for~$\mathsf{Set}$, the assignment~$\Flask_P^{\mathcal{T}}$ is a functor by the functorial properties of the pushforward. \qedhere
\end{proof} 

{We can also extend the~$\Flask$ functor to include multiple protocols by using the monad structure of~$D$. 

\begin{definition}
    Given an~$n$-tuple of protocols~$(P_1, \dots, P_n)$, the functor~$\Flask_{P_1, \dots, P_n}^{\mathcal{T}}$ is defined by induction on~$n$, with Definition~\ref{def} as the base case. The induction step of the definition is as follows.
    \begin{itemize}
        \item For a~$\mathcal{T}$-algebra~$A$, the Markov process~$\Flask_{P_1, \dots, P_n}^{\mathcal{T}}(A)$ is given by the composition \[\mathbb{N}_{\text{fs}}^{|A|} \xrightarrow{\Flask_{P_1, \dots, P_{n-1}}^{\mathcal{T}}(A)} D(\mathbb{N}_{\text{fs}}^{|A|}) \xrightarrow{D(\Flask_{P_n}^{\mathcal{T}}(A))} D^2(\mathbb{N}_{\text{fs}}^{|A|}) \xrightarrow{\mu_{\mathbb{N}_{\text{fs}}^{|A|}}} D(\mathbb{N}_{\text{fs}}^{|A|}).\]
        \item~$\Flask_{P_1,\cdots,P_n}^{\mathcal{T}}$ acts on functions identically to the~$n = 1$ case.
    \end{itemize}
\end{definition}

\begin{corollary}
$\Flask_{P_1,\cdots ,P_n}^{\mathcal{T}}$ is a well-defined functor. 
\end{corollary}
\begin{proof}
    The proof is by induction on~$n$, the base case being Theorem~\ref{prop1}. The induction step follows from naturality of~$\mu$, which ensures that the outer square commutes in the diagram below, since the left and middle squares commute by the induction assumption and Theorem~\ref{prop1} respectively.
\[\begin{tikzcd}
	{\mathbb{N}_{\text{fs}}^{|A|}} && {D(\mathbb{N}_{\text{fs}}^{|A|})} && {D^2(\mathbb{N}_{\text{fs}}^{|A|})} && {D(\mathbb{N}_{\text{fs}}^{|A|})} \\
	\\
	{\mathbb{N}_{\text{fs}}^{|B|}} && {D(\mathbb{N}_{\text{fs}}^{|B|})} && {D^2(\mathbb{N}_{\text{fs}}^{|B|})} && {D(\mathbb{N}_{\text{fs}}^{|B|})}
	\arrow["{\operatorname{Flask}_{P_1, \dots, P_{n-1}}^{\mathcal{T}}(A)}", from=1-1, to=1-3]
	\arrow["{(f_X)_*}"', from=1-1, to=3-1]
	\arrow["{D(\operatorname{Flask}_{P_n}^{\mathcal{T}}(A))}", from=1-3, to=1-5]
	\arrow["{D((f_X)_*)}"', from=1-3, to=3-3]
	\arrow["{\mu_{\mathbb{N}_{\text{fs}}^{|A|}}}", from=1-5, to=1-7]
	\arrow["{D^2((f_X)_*)}"', from=1-5, to=3-5]
	\arrow["{D((f_X)_*)}", from=1-7, to=3-7]
	\arrow["{\operatorname{Flask}_{P_1, \dots, P_{n-1}}^{\mathcal{T}}(B)}"', from=3-1, to=3-3]
	\arrow["{D(\operatorname{Flask}_{P_n}^{\mathcal{T}}(B))}"', from=3-3, to=3-5]
	\arrow["{\mu_{\mathbb{N}_{\text{fs}}^{|B|}}}"', from=3-5, to=3-7]
\end{tikzcd}\qedhere\]
\end{proof}
}
\section{Examples}

We illustrate the Flask functor in several examples. 
\begin{example}\label{exmpl_MC0}
Let~$\mathcal{I}$ be the Lawvere theory with generic object~$X$, generated freely by a single morphism~$\mathsf{interact} : X^2 \to X$. This is the theory whose algebras model components that have some method of \emph{interaction} that produces a new component from two prior components. In general, the order of the prior components in the interaction matters. We recover MC0 as the image of a certain~$\mathcal{I}$-algebra under a certain flask functor.
\begin{itemize}
    \item Let~$P_1$ to be the protocol \[\langle \pi_1, \pi_2, \operatorname{app} \circ (\operatorname{Id}_X \times \pi_{1})\rangle : X^2 \to X^3.\]In the usual notation for reaction rules, this could be written as  \[x_1 + x_2 \to x_1 + x_2 + \mathsf{interact}(x_1,x_2).\]
    \item Let~$P_2$ be the unique `delete' protocol~$X \to 1$. 
    \item Let~$L$ be the~$\mathcal{I}$-algebra given by:
    \begin{itemize}
        \item~$L(X) = M$, the set of lambda terms;
        \item~$A(\mathsf{interact})$ is the map~$(t_1, t_2) \mapsto E(t_1(t_2))$, which computes the application and subsequent reduction of two terms under our chosen reduction scheme.
    \end{itemize}
\end{itemize}

Then the MC0 reactor is described, as a Markov process, by~$\Flask_{P_1,P_2}^\mathcal{I}(L)$.
\end{example}


{ 
Our new description of MC0 teases out three parts that constitute the whole model. The domain~$\mathcal{I}$ establishes the flavour of components we are concerned with: those that interact constructively in pairs. Examples of such components in nature include coupling polymers or sexually-reproducing organisms. Next, the protocols establish the mechanics of certain `collision events' (or simply `events' when the arity is less than 2) that might involve the interaction of pairs of components. The protocol~$P_1$ can be seen as representing something similar to reproduction since the reagents are preserved, in contrast to the coupling of polymers, in which only the new molecule is left over. Finally, the algebra~$L$ provides the semantics. It tells us what the components actually are, and what `interaction' actually means. The semantics for coupling polymers could be something akin to the concatenation of strings, whereas for reproducing organisms it could be very complicated indeed.}

{
Both the protocols and the semantics are defined in terms of the domain, but they are independent of each other. The following two examples keep the domain~$\mathcal{I}$ fixed and modify the protocols and semantics, compared to Example~\ref{exmpl_MC0}.

\begin{example}\label{exmpl_prime}
    Dittrich {\em et al.}~\cite{Dittrich2001} present the \emph{number-division-chemistry} as a paradigmatic example of a constructive implicit chemistry. The `molecules' are natural numbers greater than~$1$, and there is a reaction rule 
    \[a + b \to a + a/b\]
    whenever~$b$ divides~$a$. 

    We can construct a class of similar chemistries with~$\Flask$. Let~$\mathcal{I}$ be as in Example~\ref{exmpl_MC0}, and let~$C$ be the protocol \[\langle \pi_1, \mathsf{interact}\rangle : X^2 \to X^2.\]

    We define a left integral monoid to be a monoid~$M$ such that~$bq = b'q \implies b = b'$ for all~$b, b', q \in M$. The natural numbers greater than~$0$ are a (left) integral monoid. Each integral monoid~$M$ gives rise to an algebra~$A_M$ of~$\mathcal{I}$ by setting~$A_M(X) = M$ and \[A_M(\mathsf{interact})(a,b) = \begin{cases}a/b & \text{when } b | a \\ a & \text{otherwise.}\end{cases}\]
    The number-division-chemistry (although this time with~$1$ included as a molecule), is recovered as $\Flask_C^\mathcal{I}(A_{\mathbb{N}_{\geq 1}})$. Each monoid homomorphism~$f : M \to N$ gives rise to an~$\mathcal{I}$-algebra homomorphism~$\hat f_X : A_M \to A_N$ by~$\hat f_X = f$, and hence a morphism of Markov processes~$\Flask_C^{\mathcal{I}}(\hat f)$. 
    \end{example}}

   We finish with a simple example of model refinement working with the~$\Flask$ construction. In other words, the fact that refinement of the semantics passes to refinement of the final Markov process. Here we see the construction at work in a situation apparently far-removed from `real' chemistry. 

\begin{example}\label{exmpl_library}
    Let~$C$ be as in Example~\ref{exmpl_prime}. This protocol of the interaction domain specifies a setup for modelling simple `communication systems'. The idea is that an interaction between two components is a communication from the first to the second. When we choose an algebra~$A$ of~$\mathcal{T}$, we are choosing a set of possible states~$|A|$ of the components, and a map~$A(\mathsf{interact}) : |A| \times |A| \to |A|$ that describes how the state of a receiver of a communication should be updated based on the state of the sender of the communication. Finally,~$C$ describes what actually happens when a communication takes place: the state of the sender remains the same, and the state of the receiver is updated according to~$A(\mathsf{interact})$.

    Consider a library full of people that we wish to model as a communication system in this way. This means defining an~$\mathcal{I}$-algebra whose carrier set is the set of possible states that the people in the library may take. For example, perhaps~$A(X) = \{l, m_n, m_q\}$, representing the states `librarian', `noisy member', and `quiet member'. We then define 
    \begin{center}
    \begin{tabular}{c|cccccc}
       ~$A(\mathsf{interact})(\text{row},\text{col})$ &~$l$ &~$m_n$ &~$m_q$\\
        \hline
       ~$l$&~$l$ &~$m_q$ &~$m_q$  \\
       ~$m_n$ &~$l$ &~$m_n$ &~$m_n$  \\
       ~$m_q$ &~$l$ &~$m_n$ &~$m_q$  
        
    \end{tabular}
\end{center}
to encode the idea that, for example, when a librarian talks to a noisy member they become quiet, and when a noisy member talks to a quiet member they become noisy. The dynamics of the library, under the assumption that at each time step a random person talks to another random person, is given by the Markov process~$\Flask_C^{\mathcal{T}}(A)$. 

Now suppose we have another~$\mathcal{T}$-algebra modelling the library,~$B$, and a morphism of algebras~$f : A \to B$. Such a morphism represents the fact that~$A$ is a more fine-grained model of communication in the library, which is nonetheless consistent with~$B$. For example,~$B$ could be the model given by~$B(X) = \{l, m\}$ and~$B(\mathsf{interact})(x,y) = y$, which does not distinguish between noisy and quiet members. The morphism of~$\mathcal{T}$-algebras given by~$f_X(l) = l$ and~$f_X(m_q) = f_X(m_n) = m$ is mapped by~$\Flask_C^{\mathcal{T}}$ to a corresponding morphism of the Markov processes arising from each model.
\end{example}
\section{Future directions}
To conclude, we discuss three possible future directions for work that extends, or is at least inspired by, the above construction.

\subsection{Introducing space}
Researchers of artificial life are interested in systems that can be shown to exhibit life-like behaviour. There are many perspectives on what `life-like' could mean in these contexts; one idea, originating from Maturana and Varela in the early days of artificial life research~\cite{maturana1980autopoiesis}, is the notion of \emph{autopoiesis}. 

A useful way to think about an autopoietic system is that it comprises its components and their structural relationship to one another. Then, to be autopoietic means that these relationships give rise to processes of production that both 
\begin{itemize}
    \item produce the components that make up the system (McMullin calls this \emph{closure in production}~\cite{McMullin2004}), and
    \item maintain the structural relationship of said components in opposition to a chaotic environment (McMullin calls this \emph{closure in space}).
\end{itemize}

In computational autopoiesis, the maintenance of the structural relationships of components often takes the form of some kind of physical boundary that is continually repaired by the system~\cite{Bourgine2004,Mcmullin1997}. As it stands, the~$\Flask$ construction cannot produce systems that exhibit autopoiesis interpreted in this way, because the components do not interact in any space expressive enough to facilitate structural relationships or boundaries that could be maintained. Indeed, the original computational model given by Maturana and Varela is not expressible by the~$\Flask$ construction. 

This points to a natural next step, which is to extend our description of algebraic artificial chemistries to construct systems in which components exist and interact with one another in some \emph{space}. Importantly, the spacial relationship of components should influence their dynamics, as it does in all computational models of autopoiesis. 

\subsection{Types and logic}
As it stands, the Flask construction does not generalise the two further `minimal chemistries' proposed by Fontana and Buss: MC1 and MC2~\cite{Fontana1996}. In MC1, molecules are labelled by terms in the simply typed lambda calculus. Interaction is again interpreted by application, but now molecules may only interact when the typing restriction allows. In MC2, molecules are labelled by cut-free proofs in the proof-theory of linear logic. Interaction is modelled by forming a new proof using the cut rule, and subsequently reducing to a cut-free proof using cut elimination.

The choice of single-sorted Lawvere theories as the source of component syntax generalises MC0 in a direction that is different to the direction of the `MC\_' series. In the future, we hope to tell the story of the MC\_ series in general terms by moving away from Lawvere theories, and focusing on the categorical semantics of simply typed lambda calculus and the proof theories of various logics that admit a version of the cut elimination theorem. 
\subsection{Implementation}
An important quality of any artificial chemistry is that it admits a practical computer implementation, allowing for experimentation. Recently, Rainford {\em et al.}~\cite{Rainford2020} proposed a general programmatic framework describing and implementing artificial chemistries called `MetaChem'. 

A long-term goal of our project is to contribute similar tools to the artificial life community. In particular, to use category theoretic formalisations of algebraic artificial chemistry to underpin generic and modular code that allows for experimentation with many different flavours of model all in one place.

\bibliographystyle{eptcs}
\bibliography{references}

{\bfseries Declarations}

There are no competing interests. No generative AI was used in the preparation of this work.

\end{document}